\documentclass[10pt,conference]{IEEEtran}
\IEEEoverridecommandlockouts
\usepackage{cite}
\usepackage{amsmath,amssymb,amsfonts}
\usepackage{algorithm}
\usepackage{graphicx}
\usepackage{subcaption}
\usepackage{textcomp}
\usepackage{xcolor}
\usepackage{mathrsfs}
\usepackage{url}
\usepackage{algpseudocode}
\algrenewcommand\algorithmicrequire{\textbf{Input:}}
\algrenewcommand\algorithmicensure{\textbf{Output:}}

\usepackage{amsthm}
\newtheorem{theorem}{Theorem}
\newtheorem{lemma}{Lemma}
\newtheorem{corollary}{Corollary}

\def\BibTeX{{\rm B\kern-.05em{\sc i\kern-.025em b}\kern-.08em
    T\kern-.1667em\lower.7ex\hbox{E}\kern-.125emX}}
\renewcommand{\thesection}{\arabic{section}}  
\newcounter{OPequ}

\usepackage{mathtools}

\usepackage[normalem]{ulem}
\usepackage{scalerel}
\usepackage{tikz}
\usetikzlibrary{svg.path}

\definecolor{orcidlogocol}{HTML}{A6CE39}
\tikzset{
  orcidlogo/.pic={
    \fill[orcidlogocol] svg{M256,128c0,70.7-57.3,128-128,128C57.3,256,0,198.7,0,128C0,57.3,57.3,0,128,0C198.7,0,256,57.3,256,128z};
    \fill[white] svg{M86.3,186.2H70.9V79.1h15.4v48.4V186.2z}
                 svg{M108.9,79.1h41.6c39.6,0,57,28.3,57,53.6c0,27.5-21.5,53.6-56.8,53.6h-41.8V79.1z M124.3,172.4h24.5c34.9,0,42.9-26.5,42.9-39.7c0-21.5-13.7-39.7-43.7-39.7h-23.7V172.4z}
                 svg{M88.7,56.8c0,5.5-4.5,10.1-10.1,10.1c-5.6,0-10.1-4.6-10.1-10.1c0-5.6,4.5-10.1,10.1-10.1C84.2,46.7,88.7,51.3,88.7,56.8z};
  }
}

\newcommand\orcidicon[1]{\href{https://orcid.org/#1}{\mbox{\scalerel*{
\begin{tikzpicture}[yscale=-1,transform shape]
\pic{orcidlogo};
\end{tikzpicture}
}{|}}}}

\usepackage{hyperref}

\begin{document}

%\title{Utility-Privacy Tradeoff in Multi-Agent Systems
%{\footnotesize \textsuperscript{}}

%}
\title{Trust-Aware Adaptive Disclosure for Inference Privacy Preservation in Multi-Agent Networks
{\footnotesize \textsuperscript{}}

}
\author{
Puspanjali Ghoshal$^{*}$\orcidicon{0000-0002-2324-7001} and Tobias J. Oechtering$^{\dagger}$\\
$^{*}$R. C. Bose Centre for Cryptology and Security, Indian Statistical Institute Kolkata, India.\\
$^{\dagger}$Information Science and Engineering Department,
KTH Royal Institute of Technology, Sweden.
}
%\author{Anonymous Author(s)}

\maketitle 

\begin{abstract}
Agent based systems are increasingly deployed in information critical systems including healthcare management systems, and smart grids. In this paper, we consider a multi-agent system where each agent has a latent goal that needs to be kept hidden from observing adversaries. More specifically, this paper studies privacy-preserving consensus in networked multi-agent systems under goal inference attacks. We propose a Trust-Aware Privacy Control framework that adapts message disclosure based on the dynamic trust relationships between agents. The proposed method controls information release using a trust-dependent stochastic policy. This enables a tradeoff between consensus performance and privacy preservation. Experiments demonstrate that the proposed method reduces adversarial goal inference accuracy compared to representative baselines, while maintaining competitive consensus utility, thereby highlighting the effectiveness of trust-aware mechanisms in privacy preservation of the agents in multi-agent systems.
\end{abstract}

\begin{IEEEkeywords}
Multi-Agent System (MAS), Inference Privacy, Trust
\end{IEEEkeywords}

\section{Introduction}
Multi-Agent Systems (MASs) are deployed for distributed inference tasks in multiple domains including sensor networks, and healthcare systems. In such systems, agents collaborate and communicate locally to achieve global objectives. Classical distributed algorithms like consensus and distributed optimization assume that the locally exchanged information is fully observable and trustworthy. In most real life scenarios, the local communication exposes sensitive latent information such as goals, preferences, or intentions. Studies have shown that the communication patterns can be exploited by adversaries to infer private attributes of the agents \cite{wang2024dpconsensus, ding2023review}. 

For example, consider a network of energy producers in a smart grid setting. To maintain the stability of the grid, participants must exchange operational information and coordinate amongst themselves to meet demand requirements. However, individual production strategies, reserve capacities, or operating plans may reveal commercially sensitive information to competing participants. In such a scenario, participants are required to cooperate to achieve a common objective while limiting the amount of information disclosed to less trusted parties.

%consider a fleet of autonomous drones performing search-and-rescue in a disaster zone. Each drone coordinates its trajectory with neighbours to maintain coverage and avoid collision. Additionally, each drone has a private mission goal (like, priority target location or restricted zone avoidance strategy). Unrestricted communication improves coordination, but an external observer or compromised drone can analyze message patterns to infer these hidden mission priorities. This creates a tradeoff between coordination efficiency and goal privacy.

There is, therefore, a need of decentralised privacy preserved coordination. Existing approaches relying on  differential privacy or additive noise mechanisms~\cite{Huang2015DP, deng2024DPconsensus}, degrade system performance. Recently, trust-aware coordination has been explored, where interaction is based on trust relationships~\cite{trustsurvey2025, resilientMAS2025}. However, to the best of our knowledge, the integration of trust-aware adaptation with inference-privacy guarantees remains underexplored.

We propose a Trust-Aware Privacy Control (TAPC) mechanism that dynamically regulates the amount of information disclosed based on inter-agent trust. The key idea is to reduce information exposure to low-trust agents while preserving useful communication with trusted neighbours.

The main contributions of this work are:
\begin{itemize}
\item A trust-aware probabilistic information disclosure mechanism for privacy-preserving consensus is proposed, linking trust and information leakage.
\item A theoretical characterisation of trust-leakage relationships.
\item Experimental analysis demonstrating improved privacy-utility tradeoffs.
\end{itemize}

\section{Related Work}
Machine learning–based inference attacks, including classification and sequence modeling techniques, effectively exploit message statistics and temporal dependencies \cite{tu2021AttackonComm, ma2023PrivSecInMAS}. This motivates an interest in inference-aware communication design, which will reduce the information content of messages while preserving task performance.

Privacy preservation in distributed MAS has been studied in the context of consensus, optimization, and distributed learning. Representative works include privacy-preserving distributed optimisation using stochastic noise mechanisms \cite{PGAAMAS}, and dimension reduction mechanisms for decentralized privacy preservation\cite{ghoshal2025dynamic}. Although these methods provide quantifiable privacy guarantees, they %introduce a tradeoff between privacy and system utility. Excessive noise degrades consensus accuracy, results in slow convergence, and destabilizes coordination dynamics. Moreover, most existing solutions assume uniform privacy across all communication links, 
ignore heterogeneity in trust or interaction reliability.

Trust and reputation are two central concepts in agentic systems. Frameworks based on these two concepts improve robustness against unreliable, faulty, and malicious agents \cite{trustsurvey2025, resilientMAS2025}. Existing works based on the incorporation of trust in privacy preservation either require a central trust management authority \cite{VANET}, or have been applied in query based settings \cite{MAGPIE}. Thus, the relationship between trust and decentralized privacy leakage remains largely unexplored in literature.

In contrast to existing work, this paper integrates trust with privacy-aware message generation. The proposed framework uses trust as a privacy control variable that governs selective disclosure. This allows the system to simultaneously adapt to network heterogeneity, preserve system performance, and reduce adversarial goal inference capability.

\section{Problem Formulation}
We consider a MAS comprising of \(N\) agents represented by the set $\mathcal{A}=\{1,2,\ldots,N\}.$ The communication topology is modeled as an undirected graph $\mathcal G=(\mathcal{V},\mathcal{E}),$ where \(\mathcal{V}=\mathcal{A}\) denotes the set of agents and \(\mathcal{E}\subseteq \mathcal{V}\times\mathcal{V}\) represents communication links. Let \(\mathcal N_i=\{j:(i,j)\in\mathcal E\}\) denote the set of neighbours of agent \(i\).

Let \(x_i(t)\in \mathbb{R}^{d}\) denote the observable state of agent \(i\) at time instant \(t\). Each agent possesses a latent goal variable $g_i \in \mathcal{Y},$ where \(\mathcal{Y}=\{1,2,\ldots,K\}\) denotes a finite goal space. The latent goal may correspond to a destination, mission objective, optimization preference, or task intent. Unlike the observable state \(x_i(t)\), the goal variable \(g_i\) is private and must not be inferred by unauthorised entities.

To accomplish a cooperative task, agent \(i\) transmits a message $m_{ij}(t)$. The communication history of agent \(i\) over a time horizon \(T\) is given by

\begin{equation}
\mathcal{H}_i^T= \left\{m_{ij}(t):j\in\mathcal N_i,\;t=0,\ldots,T\right\}.
\end{equation}

\subsection{Threat Model}

We consider a passive inference adversary that observes communication exchanges and exploits temporal communication patterns to infer the hidden goals of participating agents. The adversary is assumed to possess complete knowledge of the communication protocol and network topology but cannot directly access the internal states or latent goals of the agents.

Let \(\mathcal{O}_A\subseteq \mathcal{E}\) denote the set of communication links observable to the adversary. The adversary collects the communication history

\begin{equation}
\mathcal{H}_A^T=\left\{
m_{ij}(t):(i,j)\in\mathcal{O}_A; t=0,\ldots,T\right\},
\end{equation}

\noindent and constructs an estimate of the private goal $\hat g_i = f_A(\mathcal{H}_A^T),$ where \(f_A(\cdot)\) denotes an inference model.

The inference success probability is defined as 
\begin{equation}\label{eq:leakage_metric}
    P_{\mathrm{inf}}= \Pr(\hat g_i=g_i).
\end{equation} 
A larger value of \(P_{\mathrm{inf}}\) indicates a higher degree of privacy leakage. Low-trust agents are assumed to present a higher risk of information leakage and may also act as inference adversaries.

\subsection{Trust-Aware Communication Model}

Each agent maintains a trust score toward its neighbouring agents. Let $\tau_{ij}\in[0,1]$ denote the trust level assigned by agent \(i\) to neighbouring agent \(j\). A larger value of \(\tau_{ij}\) indicates a higher degree of trust. The transmitted message is generated according to a trust-aware disclosure policy

\begin{equation}
m_{ij}(t) = \pi_i\Big(x_i(t), g_i,\tau_{ij}\Big),
\end{equation}

\noindent where $\pi_i(\cdot,\cdot,\cdot)$ adjusts the amount of information disclosed based on the trust level of the receiver.

Since the communication policy explicitly depends on the trust score \(\tau_{ij}\), the resulting communication history is also trust dependent. Accordingly, the communication history observed by an adversary can be represented as $\mathcal{H}_A^T = \mathcal{H}_A^T(\tau).$ Therefore, the amount of information leaked about the latent goal is inherently influenced by the trust profile of the network.

Specifically, agents disclose more information to highly trusted neighbours and increasingly obfuscated information to low-trust neighbours. Consequently, the communication history observed by an adversary becomes a function of the trust distribution within the network.

\subsection{Inference Privacy Metric}

The objective of the adversary is to infer latent goals from communication behavior. Let $\tau = \{\tau_{ij}\mid (i,j)\in\mathcal E\}$ denote the network trust profile. The inference privacy loss associated with agent \(i\) is defined as $L_i(\tau) = I\!\left(g_i;\mathcal H_A^T \right)$, where $I\!\left(\cdot;\cdot\right)$ represent the mutual information metric. Since the communication history depends on the trust profile, the resulting leakage is indirectly influenced by $\tau$. The average network-wide inference leakage is given by

\begin{equation}
L(\tau) = \frac{1}{N} \sum_{i=1}^{N} L_i(\tau).
\end{equation}
A smaller value of $L(\tau)$ implies stronger resistance against goal inference attacks.

Mutual information metric quantifies the statistical dependence between an agent's latent goal and the observable information, thereby capturing potential inference risks arising from communication patterns. The objective considered here is to limit the disclosure of goal-related information to untrusted agents. Mutual information characterizes average information leakage rather than worst-case leakage, providing a practical measure of inference risk. However, direct estimation of mutual information from high-dimensional communication histories is challenging in practice. Therefore, the experimental section adopts adversarial goal-classification accuracy defined in Equation (\ref{eq:leakage_metric}) as an empirical surrogate for inference leakage. A random forest \cite{RandomForest} classifier is trained to infer latent goals from observed communication histories. Lower accuracy indicates lower inference leakage. Similar attack-based evaluation methodologies are widely used in privacy-preserving systems to assess resistance against realistic inference attacks.

\subsection{Task Utility}
The agents collaborate to achieve a global objective characterized by the utility function $U = U(\mathbf{x},\mathbf{m}),$ where $\mathbf{x} = [x_1,\ldots,x_N]$ and $\mathbf{m} = [m_1,\ldots,m_N].$ Depending on the application, \(U\) may represent consensus accuracy, formation stability, tracking performance, or distributed optimization efficiency.

\subsection{Problem Statement}

The proposed formulation follows a counter-adversarial design. The objective is to design a trust-aware communication policy that minimizes communication-based goal inference while preserving cooperative task performance and communication efficiency. Let the communication overhead be defined as

\begin{equation}
C(\mathbf{m}) = \sum_{t=0}^{T}\sum_{(i,j)\in\mathcal E}|m_{ij}(t)|,
\end{equation}

\noindent where $|m_{ij}(t)|$ denotes the size of the message transmitted from agent $i$ to agent $j$ in bits. The trust-aware inference privacy preservation problem is formulated as

\begin{equation}
\begin{aligned}
\min_{\pi} \quad & L(\tau) \\
\text{s.t.}\quad & U(\mathbf{x},\mathbf{m}) \ge U_{\min}, \\
& C(\mathbf{m}) \le C_{\max}, \\
& m_{ij}(t) = \pi_i \Big(x_i(t),g_i,\tau_{ij}\Big).
\end{aligned}
\label{eq:optimization_problem}
\end{equation}

The optimization problem seeks a communication policy that minimizes the information an adversary can infer about the latent goals of agents while guaranteeing a minimum level of task utility and respecting communication constraints.

%\vspace{4mm}
\section{Trust-Aware Privacy Control Framework}\label{Relationship}

We now illustrate the proposed trust-aware framework for mitigating goal inference attacks in multi-agent systems. The framework consists of four components: trust evaluation, adaptive disclosure control, privacy-preserving message generation, and leakage-aware optimization.

\subsection{Trust Evaluation Model}

Each agent maintains trust scores for its neighbouring agents based on historical communication behavior and interaction outcomes. Let $r_{ij}(t)$ denote the communication reliability of agent $j$ as observed by agent $i$, and let $s_{ij}(t)$ denote a behavioral consistency score. The communication reliability is defined as

\begin{equation}
r_{ij}(t) = \frac{N_{ij}^{\mathrm{succ}}(t)}{N_{ij}^{\mathrm{tot}}(t)},
\label{eq:reliability}
\end{equation}

\noindent where $N_{ij}^{\mathrm{succ}}(t)$ and $N_{ij}^{\mathrm{tot}}(t)$ denote the number of successful and total communication interactions, respectively, within a predefined observation window. To quantify behavioral consistency, each agent compares recent communication behavior with historical patterns. The consistency score is defined as

\begin{equation}
s_{ij}(t) = \exp\left(-\frac{1}{W}\sum_{k=1}^{W}\left|m_{ij}(t-k)-\bar m_{ij}(t)\right| \right),
\label{eq:consistency}
\end{equation}

\noindent where $W$ denotes the observation window and $\bar m_{ij}(t)$ denotes the average communication behavior. The trust score assigned by agent $i$ to agent $j$ is computed as

\begin{equation}
\tau_{ij}(t) = \lambda r_{ij}(t) + (1-\lambda)s_{ij}(t),
\label{eq:trust}
\end{equation}

\noindent where $0\le\lambda\le1$ controls the relative importance of reliability and behavioral consistency. Since both reliability and consistency lie within $[0,1]$, the resulting trust score also satisfies $\tau_{ij}(t)\in[0,1]$. Higher trust scores indicate a more reliable communication partner, whereas lower values indicate potentially risky agents that may facilitate privacy leakage. The network trust profile is represented by $\tau= \left\{\tau_{ij}(t)\;|\;(i,j)\in\mathcal E \right\}.$

\subsection{Trust-Based Disclosure Control}

To regulate information sharing, each communication link is assigned a disclosure coefficient. The disclosure coefficient associated with communication link $(i,j)$ is defined as

\begin{equation}
\alpha_{ij}(t) = \tau_{ij}^{\gamma}(t),
\label{eq:alpha}
\end{equation}

\noindent where $\gamma>0$ is a sensitivity parameter.

The coefficient satisfies $0\le\alpha_{ij}(t)\le1.$ When $\alpha_{ij}(t)$ approaches one, a larger amount of task-relevant information is disclosed. Conversely, smaller values correspond to stronger privacy protection.

\subsection{Trust-Aware Message Generation}

The key idea of the proposed framework is to selectively disclose task-relevant information while limiting the amount of goal-related information exposed to potentially untrusted agents. Let $z_i(t) = \phi\left(x_i(t),g_i\right)$ denote a task-relevant feature extracted from the observable state and latent goal of agent $i$, where $\phi(\cdot)$ is an application-dependent feature mapping. The transmitted message is generated as

\begin{equation}
m_{ij}(t) = \alpha_{ij}(t) z_i(t) + \left(1-\alpha_{ij}(t)\right)\eta_i(t),
\label{eq:message}
\end{equation}

\noindent where $\eta_i(t)$ denotes an obfuscation signal. In this work, Gaussian perturbation is adopted, $\eta_i(t) \sim \mathcal N \left(0,\sigma^2 I \right),$ where $\sigma^2$ denotes the noise variance, resulting in a trust-dependent Gaussian disclosure mechanism. The focus of the proposed framework is the adaptive trust dependent disclosure policy. Alternative perturbation mechanisms may also be incorporated. Equation (\ref{eq:message}) ensures that highly trusted neighbours receive more informative messages, while low-trust neighbours receive increasingly obfuscated information.

% \subsection{Inference Leakage Analysis}

% The adversary observes the communication history $\mathcal H_A^T(\tau).$ The inference privacy leakage associated with agent $i$ is quantified using conditional mutual information:

% \begin{equation}
% L_i(\tau) = I \left(g_i;\mathcal H_A^T \mid \tau \right),
% \label{eq:mi}
% \end{equation}

% where $I(\cdot;\cdot)$ denotes mutual information.

% The average network-wide inference leakage is given by

% \begin{equation}
% L(\tau) = \frac{1}{N} \sum_{i=1}^{N} L_i(\tau).
% \label{eq:network_leakage}
% \end{equation}

% A smaller value of $L(\tau)$ implies stronger resistance against communication-based goal inference attacks.

\subsection{Leakage Upper Bound}

The disclosure coefficient directly affects the signal-to-noise ratio available to the adversary. Let

\begin{equation}
P_z = \mathbb E\left[|z_i(t)|^2\right]
\end{equation}

\noindent denote the average power of the task-relevant feature. The effective signal-to-noise ratio observed on communication link $(i,j)$ is

\begin{equation}
\mathrm{SNR}_{ij} = \frac{ \alpha_{ij}^{2} P_z}{\left(1-\alpha_{ij}\right)^2 \sigma^2}.
\label{eq:snr}
\end{equation}

Under the proposed message-generation model, the adversarial observation process can be approximated as an additive white Gaussian noise (AWGN) channel with effective signal-to-noise ratio given by Equation (\ref{eq:snr}). Thus, the mutual information between the transmitted signal and the adversarial observation is upper bounded by the Shannon channel capacity formula \cite{Cover}. Therefore,

\begin{equation}
L_i(\tau) \le \frac{1}{2}\log \left(1+\mathrm{SNR}_{ij}\right).
\label{eq:bound}
\end{equation}
\noindent with equality achieved when the effective transmitted feature is Gaussian distributed. This equation establishes an explicit relationship between trust-aware disclosure and inference leakage. Since the latent goal influences the transmitted feature $z_i(t)$, the information available for goal inference cannot exceed the information contained in the observed communication channel. The complete algorithm is outlined in Algorithm \ref{alg:trust}.

\begin{algorithm}[h]
\caption{Trust-Aware Privacy-Preserving Communication}
\label{alg:trust}
\begin{algorithmic}[1]

\For{each communication round $t$}
    \For{each link $(i,j)\in\mathcal E$}
        \State Compute reliability $r_{ij}(t)$
        \State Compute consistency score $s_{ij}(t)$
        \State Update trust score using (\ref{eq:trust})
        \State Compute disclosure coefficient using (\ref{eq:alpha})
        \State Generate noise $\eta_i(t)$
        \State Construct message using (\ref{eq:message})
        \State Transmit $m_{ij}(t)$
    \EndFor
\EndFor

\end{algorithmic}
\end{algorithm}

The computational complexity is dominated by trust updates and message generation across communication links. Therefore, the overall complexity per communication round is $O(|\mathcal E|).$

\section{Theoretical Analysis}\label{theo-Results}
% \textbf{Proposition 1:}
% For a fixed noise variance $\sigma^2$, the inference leakage upper bound is a monotonically increasing function of the disclosure coefficient $\alpha_{ij}$.

% \textit{Proof:} From (\ref{eq:snr}),

% \begin{equation}
% \frac{\partial \mathrm{SNR}_{ij}} {\partial \alpha_{ij}} >0,\qquad 0<\alpha_{ij}<1.
% \end{equation}

% Furthermore, the function

% \begin{equation}
% f(x) = \frac{1}{2}\log(1+x)
% \end{equation}
% is strictly increasing for $x\ge0$.

% Therefore,

% \begin{equation}
% \frac{\partial L_i}{\partial \alpha_{ij}} > 0.
% \end{equation}

% Hence, increasing the disclosure coefficient increases inference leakage, while reducing disclosure decreases the information available to the adversary.

% \bigskip

\begin{lemma}
    For fixed noise variance $\sigma^2$, the inference leakage upper bound is monotonically increasing with the disclosure coefficient $\alpha_{ij}$.
\end{lemma}
\begin{proof}
Differentiating Equation (\ref{eq:snr}) with respect to $\alpha_{ij}$ yields

\begin{equation*}
\frac{\partial \mathrm{SNR}_{ij}}{\partial \alpha_{ij}} = \frac{2\alpha_{ij}P_z} {\sigma^2(1-\alpha_{ij})^{3}}.
\end{equation*}

Since $P_z>0$, $\sigma^2>0$, and $0<\alpha_{ij}<1$, it follows that

\begin{equation*}
\frac{\partial \mathrm{SNR}_{ij}}{\partial \alpha_{ij}}>0.
\end{equation*}

Furthermore, the function $f(x)=\frac{1}{2}\log(1+x)$ is strictly increasing for $x\ge0$. Therefore,

\begin{equation*}
\frac{\partial L_i}
{\partial \alpha_{ij}}
=
\frac{\partial L_i}
{\partial \mathrm{SNR}_{ij}}
\frac{\partial \mathrm{SNR}_{ij}}
{\partial \alpha_{ij}}
>0.
\end{equation*}

Hence, the leakage upper bound increases monotonically with the disclosure coefficient $\alpha_{ij}$. Consequently, reducing disclosure lowers the information available to the adversary and decreases inference leakage.
\end{proof}

% \begin{theorem}(Privacy-Utility Tradeoff)
%     Assume that the task utility function $U(\mathbf{x},\mathbf{m})$ is continuous and monotonically increasing with respect to the disclosure coefficients $\alpha_{ij}$. Then there exists at least one Pareto-optimal disclosure vector $\boldsymbol{\alpha}^{\star}$ that balances task utility and inference privacy.
% \end{theorem}
% \begin{proof}
%     From Lemma 1, the inference leakage $L(\tau)$ is monotonically increasing with the disclosure coefficients. By assumption, $U(\mathbf{x},\mathbf{m})$ is also monotonically increasing with $\alpha_{ij}$. Therefore, improving privacy by reducing $\alpha_{ij}$ necessarily reduces utility, while increasing utility requires larger disclosure. Hence, privacy and utility constitute conflicting objectives.

%     The feasible set 
%     \begin{equation}
%         \mathcal A= \left\{ \boldsymbol{\alpha} : 0\le \alpha_{ij}\le1, \; (i,j)\in\mathcal E \right\}
%     \end{equation}
%     is compact and nonempty. Since both objective functions are continuous on $\mathcal A$, the existence of efficient solutions follows from standard multi-objective optimization theory. Therefore, at least one Pareto-optimal disclosure vector $\boldsymbol{\alpha}^{\star}$ exists that achieves a balance between inference privacy and task utility.
% \end{proof}
\begin{theorem}(Trust-Leakage Relationship)
    For a fixed noise variance $\sigma^2$ and sensitivity parameter $\gamma>0$, the inference leakage upper bound is a monotonically increasing function of the trust score $\tau_{ij}$.
\end{theorem}
\begin{proof}
    Differentiating (\ref{eq:alpha}) with respect to $\tau_{ij}$ gives

\begin{equation}
\frac{\partial \alpha_{ij}}{\partial \tau_{ij}} = \gamma\tau_{ij}^{\gamma-1}.
\end{equation}

Since $\gamma>0$ and $\tau_{ij}\in(0,1]$,

\begin{equation}
\frac{\partial \alpha_{ij}}{\partial \tau_{ij}}>0.
\end{equation}

Using the chain rule together with Lemma~1,

\begin{equation}
\frac{\partial L_i}{\partial \tau_{ij}} = \frac{\partial L_i}{\partial \alpha_{ij}}\frac{\partial \alpha_{ij}} {\partial \tau_{ij}} >0.
\end{equation}

Therefore, the inference leakage upper bound increases monotonically with the trust score $\tau_{ij}$. Consequently, highly trusted communication links reveal more information to the receiver and potentially to an observing adversary. Thus, this theorem proves that for a specific communication link, increasing trust increases information disclosure and therefore leakage on that link.
\end{proof}

\begin{corollary}(Effect of Sensitivity Parameter)
    For a fixed trust score $\tau_{ij}\in(0,1)$, increasing the sensitivity parameter $\gamma$ reduces the inference leakage upper bound.
\end{corollary}
\begin{proof}
    Differentiating (\ref{eq:alpha}) with respect to $\gamma$ yields

\begin{equation}
\frac{\partial \alpha_{ij}}{\partial \gamma} = \tau_{ij}^{\gamma}\ln(\tau_{ij}).
\end{equation}

Since $0<\tau_{ij}<1$, we have $\ln(\tau_{ij})<0$, which implies

\begin{equation}
\frac{\partial \alpha_{ij}}{\partial \gamma}<0.
\end{equation}

Using Lemma~1,

\begin{equation}
\frac{\partial L_i}{\partial \gamma} = \frac{\partial L_i}{\partial \alpha_{ij}}\frac{\partial \alpha_{ij}}{\partial \gamma} < 0.
\end{equation}

Hence, increasing the sensitivity parameter $\gamma$ decreases the inference leakage upper bound. Therefore, $\gamma$ acts as a tunable privacy-control parameter that allows stronger privacy protection without modifying the trust evaluation mechanism.
\end{proof}

For the preliminary theoretical analysis of the proposed trust dependent framework, the leakage analysis has been carried out by considering a simplified instantaneous setting. Alternative message generation mechanisms along with the incorporation of dynamic temporal message history will be explored in future extensions.

\section{Experiments and Results}\label{exp-Results}

This section evaluates the proposed Trust-Aware Privacy Control (TAPC) framework against three representative baseline methods. Performance is evaluated in terms of consensus utility, inference leakage, and communication cost.

\subsection{Experimental Setup}

We consider a network of $N=50$ agents interacting over a connected random geometric graph. Each agent is assigned a latent goal from a discrete set of size $K=6$. The system evolves over $T=100$ communication rounds. To assess privacy, the experimental evaluation employs an adversary that observes communication histories and performs goal inference using machine learning.

The TAPC method introduces trust-dependent stochastic perturbation in message generation, while the baselines represent full disclosure (FD), uniform Gaussian noise (GN), and trust-thresholded communication (TOC).

In the context of privacy-preserving consensus~\cite{wang2024dpconsensus, ding2023review}, FD corresponds to the standard baseline where agents exchange unperturbed state information. Noisy communication mechanisms, represented in this work as GN, have been widely adopted for privacy in distributed settings~\cite{Huang2015DP}. In contrast, TOC builds on trust-aware communication strategies where information is released only to neighbours that satisfy a prespecified trust threshold~\cite{trustsurvey2025,resilientMAS2025}.

%Inference leakage is measured as the classification accuracy of a Random Forest adversary trained on statistical features extracted from temporal communication histories. The attacker uses aggregated message statistics to infer latent agent goals.

\subsection{Consensus Utility}

Consensus performance is measured using:
\begin{equation}
U = \frac{1}{1 + \frac{1}{N} \sum_{i=1}^{N} \|x_i - \bar{x}\|},
\end{equation}
where $\bar{x}$ denotes the mean state. As shown in Figure~\ref{fig:utility}, FD achieves the highest utility due to unrestricted communication. GN significantly degrades performance due to injected noise. TOC improves utility by selectively preserving structured communication on high-trust links. TAPC maintains competitive utility while incorporating stochastic privacy protection.

\begin{figure}[h]
\centering
\includegraphics[width=0.48\textwidth]{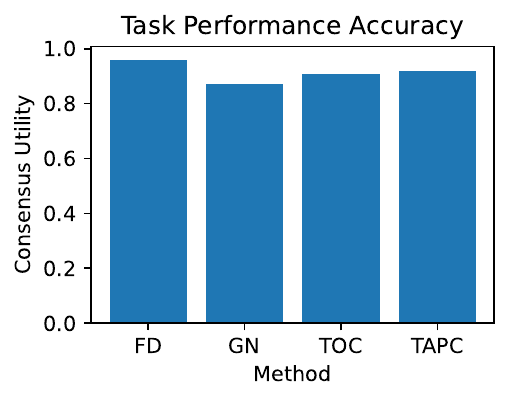}
\caption{Consensus utility comparison across methods.}
\label{fig:utility}
\end{figure}

\subsection{Inference Leakage Analysis}

Inference leakage is computed using Equation (\ref{eq:leakage_metric}) under a random forest adversary \cite{RandomForest} trained on statistical features extracted from temporal communication histories. As shown in Figure~\ref{fig:leakage}, FD results in complete information leakage, while GN reduces leakage at the cost of degraded utility. TOC further reduces leakage by limiting communication on low-trust links. TAPC achieves the lowest leakage among all evaluated methods, demonstrating the effectiveness of trust-aware stochastic masking.

\begin{figure}[h]
\centering
\includegraphics[width=0.48\textwidth]{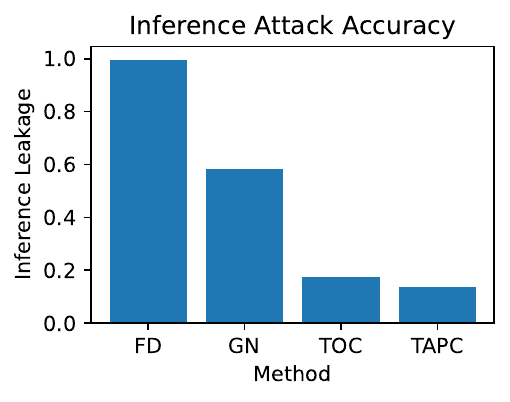}
\caption{Inference leakage comparison across methods.}
\label{fig:leakage}
\end{figure}

\subsection{Privacy--Utility Pareto chart}

Figure~\ref{fig:tradeoff} illustrates the tradeoff between utility and leakage. FD lies at the high-utility but high-leakage extreme, while GN shifts toward lower leakage but reduced performance. TOC provides a balanced intermediate tradeoff, whereas TAPC achieves a strong Pareto-efficient operating point among the evaluated methods.

\begin{figure}[h]
\centering
\includegraphics[width=0.48\textwidth]{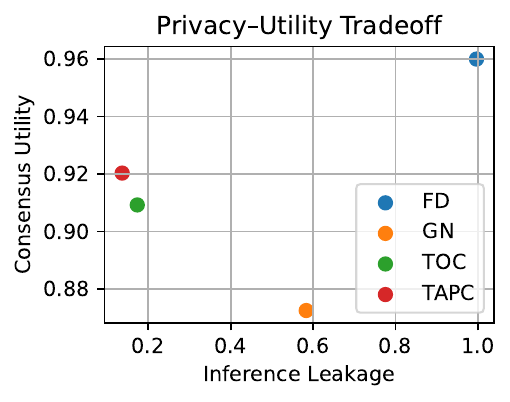}
\caption{Privacy-utility tradeoff between inference leakage and consensus utility.}
\label{fig:tradeoff}
\end{figure}

\subsection{Communication Cost}

With the number of agents fixed, all methods incur identical communication cost since each agent communicates over all edges at every iteration, isolating the effect of message design. The proposed TAPC method exhibits a steady increase in the communication cost as the number of agents increases. This is depicted in Figure \ref{fig:scalability}.

\begin{figure}[h]
\centering
\includegraphics[width=0.48\textwidth]{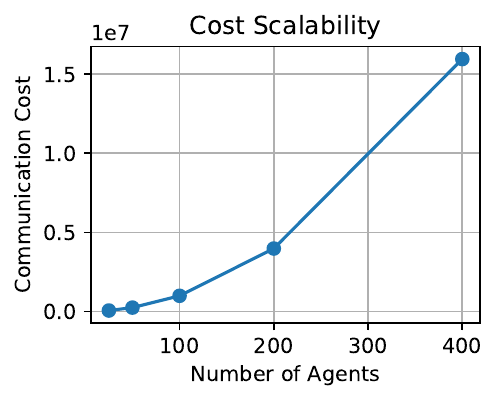}
\caption{Cost scaling with respect to number of agents.}
\label{fig:scalability}
\end{figure}

\subsection{Results Summary}

Table~\ref{tab:results} summarizes the performance for 50 agents. The experimental results demonstrate that TAPC achieves a strong balance between privacy and utility, consistently reducing inference leakage while maintaining competitive consensus performance compared to baseline methods.

\begin{table}[h]
\centering

\begin{tabular}{|p{0.2\linewidth}|p{0.2\linewidth}|p{0.2\linewidth}|p{0.2\linewidth}|}
\hline
Method & Utility & Leakage & Cost \\
\hline
FD   & 0.960 & 0.997 & 245000 \\
GN   & 0.872 & 0.583 & 245000 \\
TOC  & 0.909 & 0.173 & 245000 \\
TAPC & 0.920 & 0.137 & 245000 \\
\hline
\end{tabular}\caption{Performance comparison of different methods}
\label{tab:results}
\end{table}

%\vspace{4mm}
\section{Conclusion and Future Work}\label{Future}
This paper proposes a Trust-Aware Privacy Control (TAPC) framework for multi-agent consensus systems under inference-based privacy threats. Unlike conventional approaches that rely on uniform noise injection or hard trust thresholding, the proposed method introduces a trust-dependent stochastic masking mechanism that adaptively regulates information disclosure during inter-agent communication.

Experimental evaluations demonstrate the efficiency of TAPC over representative baselines. TAPC achieves a balance between consensus performance and privacy preservation. The results highlight that adaptive trust-aware stochastic control is more effective than both uniform perturbation and deterministic trust filtering for mitigating information leakage in multi-agent networks. This suggests that incorporating trust structure into privacy mechanisms provides a principled direction for improving secure coordination in distributed systems. The current theoretical analysis, however, considers a simplified instantaneous communication setting and establishes explicit relationships between trust, disclosure, and inference leakage.

Future work will extend the framework to account for temporal dependencies and information aggregation across communication histories. We will further extend the proposed framework to time-varying, directed communication graphs. This would enable its application to dynamic networked systems. Second, the message generation framework will be changed from the preliminary Gaussian noise to alternative mechanisms and the theoretical guarantees on privacy leakage bounds and convergence rates would be strengthened appropriately. Finally, TAPC would be deployed in real-world multi-agent applications to further validate its practical effectiveness.

\bibliographystyle{plain} % We choose the "plain" reference style
\bibliography{refs}
\vspace{12pt}

\end{document}